\documentclass[11pt]{article}
\usepackage{amssymb, latexsym, mathtools, amsthm}
\begingroup
    \makeatletter
    \@for\theoremstyle:=definition,remark,plain\do{%
        \expandafter\g@addto@macro\csname th@\theoremstyle\endcsname{%
            \addtolength\thm@preskip\parskip
            }%
        }
\endgroup
\usepackage{enumerate}
\usepackage{pgf,tikz}
\usepackage{pdfsync}
\usepackage{txfonts}
\usepackage[T1]{fontenc}
\usepackage{booktabs}
 \usepackage{graphicx}
\definecolor{dnrbl}{rgb}{0,0,0.3}
\definecolor{dnrgr}{rgb}{0,0.3,0}
\definecolor{dnrre}{rgb}{0.5,0,0}
\usepackage[colorlinks=true, citecolor=dnrgr, linkcolor=dnrre, urlcolor=dnrre] {hyperref}
\usepackage{xcolor}
\usepackage{parskip}
\usepackage{tabularx, colortbl}
\usepackage{geometry}
 \usepackage[scriptsize, up]{caption}
\usepackage{pgf,tikz}
\usetikzlibrary{decorations, decorations.pathmorphing}
\usetikzlibrary {shapes}

\theoremstyle{plain}
\newtheorem{thm}{Theorem}[section]

\newtheorem{lem}[thm]{Lemma}

\newtheorem{prob}[thm]{Problem}
\newtheorem{defi}[thm]{Definition}
\makeatletter

\let\c@table\c@figure
\makeatother

\newcommand{\omux}{$\Omega_U(X)$ }
\newcommand{\omuxn}{$\Omega_U(X)$}
\newcommand{\omuxf}{\Omega_U(X)}

\newcommand{\Nat}{\mathbb{N}}

\newcommand{\de}{\downarrow} %defined

\newcommand{\DD}{\mathcal{D}}
\newcommand{\QQ}{\mathbb{Q}^+}

\newcommand{\FSW}{Figueira, Stephan, and Wu\ } 
\newcommand{\MN}{Miller and Nies\ }

\newcommand{\BFGM}{Becher, Figueira, Grigorieff, and Miller\ }
\newcommand{\KS}{Ku{\v{c}}era and Slaman\ }
\newcommand{\CHKW}{Calude, Hertling, Khoussainov and Wang\ }
\newcommand{\DHN}{Downey, Hirschfeldt and Nies\ }
\renewcommand{\DH}{Downey and Hirschfeldt\ }
\newcommand{\CC}{\mathcal{C}}
\newcommand{\DHNS}{Downey, Hirschfeldt, Nies and Stephan } 

\newcommand{\ml}{Martin-L\"{o}f }
\newcommand{\pz}{$\Pi^0_1$\ }

\newcommand{\ie}{i.e.\ }
\newcommand{\ce}{c.e.\ }
\newcommand{\dce}{d.c.e.\ }
\newcommand{\lce}{left-c.e.\ }
\newcommand{\rce}{right-c.e.\ }

\newcommand{\pf}{prefix-free }

\renewenvironment{abstract}
 { \normalsize
  \list{}{
    \setlength{\leftmargin}{.0cm}%
    \setlength{\rightmargin}{\leftmargin}%
    }%
  \item {\bf \abstractname.} \relax}
 {\endlist}

\title{Differences of halting probabilities
\thanks{Barmpalias was supported by the 
1000 Talents Program for Young Scholars from the Chinese Government, grant no.\ D1101130.
Additional support was received by
the Chinese Academy of Sciences (CAS) and the Institute of Software of the CAS.
Lewis-Pye was supported by a Royal Society University 
Research Fellowship.}}

\author{George Barmpalias  \and Andrew Lewis-Pye}
\date{\today}
\begin{document}
\maketitle
\begin{abstract}
The halting probabilities of universal \pf machines are universal for the class of reals with
computably enumerable left cut (also known as \lce reals), 
and coincide with the \ml random elements of this class.
We study the differences of  \ml random \lce reals and show that for each pair of such reals 
$\alpha,\beta$ there exists a unique number $r>0$ such that 
$q\alpha-\beta$ is a \ml random \lce real
for each positive rational $q>r$  and a \ml random 
\rce real for each positive rational $q<r$.
Based on this result we develop a theory of differences of halting probabilities, which answers a number
of questions about \ml random \lce reals, 
including one of the few remaining open problems from the list
of open questions in algorithmic randomness \cite{MR2248590}.

The halting probability of a prefix-free machine $M$ restricted to a set $X$
is the probability that the machine halts and outputs an element of $X$.
These numbers $\Omega_M(X)$ were studied in
\cite{DBLP:journals/jsyml/BecherG05,jsyml/BecherFGM06,tcs/BecherG07,jsyml/BecherG09}
as a way to obtain concrete highly random numbers.
When $X$ is a \pz set, the number $\Omega_M(X)$ 
is the difference of two halting probabilities.
\BFGM asked whether $\Omega_U(X)$ is \ml random when $U$ is 
universal and $X$ is a \pz set. This problem has resisted numerous attempts 
\cite{DBLP:journals/jsyml/BecherG05,jsyml/BecherFGM06,jc/FigueiraSW06}.
We apply our theory of differences of halting probabilities to 
give a positive answer,  and show that $\Omega_U(X)$ is a \ml random \lce real
whenever $X$ is a nonempty \pz set.
\end{abstract}
\vspace*{\fill}
\noindent{\bf George Barmpalias}\\[0.5em]
\noindent
State Key Lab of Computer Science, 
Institute of Software, Chinese Academy of Sciences, Beijing, China.
School of Mathematics, Statistics and Operations Research,
Victoria University of Wellington, New Zealand.\\[0.2em] 
\textit{E-mail:} \texttt{\textcolor{dnrgr}{barmpalias@gmail.com}}.
\textit{Web:} \texttt{\textcolor{dnrre}{http://barmpalias.net}}\par
\addvspace{\medskipamount}\medskip\medskip
\noindent{\bf Andrew Lewis-Pye}\\[0.5em]  
\noindent Department of Mathematics,
Columbia House, London School of Economics, 
Houghton Street, London, WC2A 2AE, United Kingdom.\\[0.2em]
\textit{E-mail:} \texttt{\textcolor{dnrgr}{A.Lewis7@lse.ac.uk.}}
\textit{Web:} \texttt{\textcolor{dnrre}{http://aemlewis.co.uk}} 

\vfill \thispagestyle{empty}
\clearpage

\section{Introduction}
Perhaps the most recognizable algorithmically random number is Chaitin's $\Omega$ number.
This is the probability that a universal prefix-free machine  $U$ halts 
when we feed the input with successive bits (zeros and ones) until a computation converges,
and is usually denoted by  $\Omega_U$.
The fact that the underlying machine has a prefix-free domain allows for this number
to be defined with the following rather simple formula: 
\begin{equation}\label{eq:Uhaltprob}
\Omega_U=\sum_{U(\sigma)\de} 2^{-|\sigma|}.
\end{equation}
Prefix-free machines are Turing machines whose domain is a
prefix-free subset of the finite binary strings, so that they operate
instantaneous codes (\ie uniquely decodable without out-of-band markers).
Such machines are essentially equivalent to self-delimiting machines, \ie Turing machines
with a one-way input and one-way output tape, such that a convergent computation on a
finite binary input $\sigma$
cannot be extended to a different computation on an input which extends $\sigma$.
Chaitin \cite{MR0411829} used these machines in order to give a definition of
randomness for infinite sequences in terms of incompressibility, and 
showed that for each universal 
prefix-free machine $U$ the number
$\Omega_U$ is algorithmically random, in the sense that its binary expansion is a
random sequence according to the definition of \ml \cite{MR0223179}.
The probability that
$U$ halts and outputs a string in a set $X$
of binary strings is:
\begin{equation}\label{eq:UhaltproX}
\Omega_U(X)=\sum_{U(\sigma)\de\in X} 2^{-|\sigma|}.
\end{equation}
 Chaitin \cite{chaitin2004algorithmic} observed that  
 if $X$ is a computably enumerable set, 
then $\Omega_U(X)$ is \ml random. The question as to whether there is a set $X$ such that
its complement is computably enumerable (\ie $X$ is a $\Pi^0_1$ set) and 
$\Omega_U(X)$  {\em is not} \ml random
was asked and discussed by a number of authors 
\cite{DBLP:journals/jsyml/BecherG05,jsyml/BecherFGM06,jc/FigueiraSW06},
occasionally along with partial solutions. 
It is also one of the last remaining problems (Question 8.10)
in the list of open
problems in algorithmic randomness by \MN  \cite{MR2248590}. 
 There is a considerable background and an original motivation surrounding this question
of restricted halting probability, 
which we defer to Section \ref{subse:histcpro} in order to focus on our present contribution.
\begin{prob}[Question 8.10 in \MN \cite{MR2248590}]\label{3aCBk84Vgj}
If $U$ is a universal machine and $X\neq\emptyset$ is a \pz set, is the probability
$\Omega_U(X)$ always a \ml random number?
\end{prob}
This open problem was the starting point for our investigations, 
which led to the study of more general questions concerning the differences of $\Omega$ numbers, and
revealed a missing theory which is complementary to the well developed theory of halting probabilities
(see \DH \cite[Chapter 9]{rodenisbook} for an overview). 
Before we present our solution and, perhaps more interestingly, 
the intriguing theory of differences of halting probabilities that it inspired, we make our discussion precise by
giving a formal definition of universality. Informally, a Turing machine is universal if it can simulate any other
Turing machine.
\begin{defi}[Universal \pf machines]
Given an effective list $(M_e)$ of all \pf machines, a \pf machine $U$ is universal if 
there exists a computable function $e\mapsto\sigma_e$ such that, for all $\tau\in 2^{<\omega}$, 
$U(\sigma_e\ast\tau)\simeq M_e(\tau)$.
\end{defi}
As usual, the symbol  $\simeq$ denotes that either both sides of the relation are defined and equal, or both sides
are undefined. In Kolmogorov complexity theory, universal machines are sometimes called 
{\em universal by adjunction} in order to distinguish them from a wider class of machines, the
{\em optimal \pf machines}. The latter class consists of the \pf machines with 
respect to which the Kolmogorov complexity function
is minimal modulo an additive constant, within the class of all \pf machines.
In their attempt at Problem \ref{3aCBk84Vgj},
\FSW \cite{jc/FigueiraSW06} constructed a special optimal \pf machine 
$U$ and a $\Pi^0_1$ set $X$ such that 
 $\Omega_U(X)$ is not \ml random.
However their machine is not universal and, 
as will become clear in the following, this approach has
little to do with a solution to the problem. 
\BFGM showed in \cite{DBLP:journals/jsyml/BecherG05} that, given a universal \pf machine $U$,
 there is a $\Delta^0_2$
set $X$ such that $\Omega_U(X)$ is not \ml random.
\subsection{Solution to the problem of restricted halting probability}
Despite these negative results, we give a rather surprising positive answer to
Problem \ref{3aCBk84Vgj}. Recall that a real is \lce if it has a computably enumerable (in short, c.e.) 
left Dedekind cut or, equivalently, if it is the limit of a computable
increasing sequence of rationals. Moreover we point out that \lce reals 
can be seen as halting probabilities of \pf machines, if one considers the Kraft-Chaitin algorithm for 
constructing \pf machines (see e.g.\ \cite[Section 2.6]{rodenisbook}). Conversely, every
halting probability of a \pf machine is a \lce real, so the two classes of reals coincide.

\begin{thm}\label{JgIfqiIPd}
If $U$ is a universal \pf machine and $X$ is a nonempty \pz set, the number $\Omega_U(X)$ is
a \ml random \lce real.
\end{thm}
Let us illustrate the novelty of this result with an example. 
Given a $\Sigma^0_1$ set $X$, there is a machine which enumerates $X$, and since
the convergent computations of a \pf machine $U$ are also computably enumerable,
it is not surprising that $\Omega_U(X)$ has a computably enumerable Dedekind cut.
As the enumeration of $X$ progresses, more programs halt on elements of $X$, 
and the probability of this happening according to \eqref{eq:UhaltproX} can be approximated by
a computable increasing sequence of rationals.
If, however, $X$ is a \pz set, 
such as the non-theorems of Peano arithmetic, there is no apparent way to approximate
$\Omega_U(X)$ without overestimating its value. Indeed, a correct approximation would have to account
for the programs that converge to 
the non-theorems of Peano arithmetic. As one enumerates the theorems of Peano arithmetic, the non-computability of this set means that there is no way in general to be sure that a given statement will not be enumerated at a later stage, at which point any previous consideration
of programs converging to that statement will have introduced a possible overestimation of
$\Omega_U(X)$.
Despite this, not only does there exist a way to approximate 
$\Omega_U(X)$ without ever overestimating it, but any computable approximation to it
essentially already has this property:
\begin{equation}\label{A69imMhIZ}
\parbox{14cm}{Given any universal \pf machine $U$,
any \pz set $X\neq\emptyset$ and any computable sequence of rationals $(\alpha_s)$
converging to $\Omega_U(X)$ we have $\alpha_s<\Omega_U(X)$ for all but finitely many $s$.}
\end{equation}
This fact will be derived from Theorem \ref{JgIfqiIPd} and the properties of \ml random 
\lce reals. So what does Problem \ref{3aCBk84Vgj} have to do with differences of halting probabilities?
It is not hard to see that  $\Omega_U(X)$ can be written as  $\Omega_U-\Omega_U(\overline{X})$, where
$\overline{X}$ denotes the complement of $X$. Moreover if $X$ is \pz then $\overline{X}$ is \ce so
$\Omega_U(\overline{X})$ is a \ml random \lce real. By central results from the theory of 
\ml random \lce reals (see Section \ref{MXWRrIRDRZ}), \ml random \lce reals are exactly 
the halting probabilities of universal \pf machines. 
Therefore there exists a universal \pf machine $V$ such that
$\Omega_V=\Omega_U(\overline{X})$ and hence, $\Omega_U(X)$ is the difference 
$\Omega_U-\Omega_V$ of two halting probabilities of universal \pf machines.
So the question as to whether $\Omega_U(X)$ is \ml random has much to do with
what happens if we subtract two \ml random \lce reals.

Reals that are differences of \lce reals (i.e.\ of the form $\alpha-\beta$ for $\alpha,\beta$ left-c.e.) are often called \dce reals. 
There has been considerable work on this class of reals 
\cite{Ambos.ea:00,Raichev:05,Ng06,mlq/DowneyWZ04,mlq/DowneyWZ04}. 
Moreover the theory of \ml random \lce reals is well understood, initially with 
\cite{Solovay:75,Calude.Hertling.ea:01,Kucera.Slaman:01} and on a deeper level
with \cite{Downey02randomness}. Despite this considerable body of work, the
basic question as to what happens when we subtract two \ml random \lce reals has remained largely untouched. 
It turns out that the solution of Problem \ref{3aCBk84Vgj} crucially depends
on the answer to this question, which we present in the next section.

\subsection{A theory of differences of universal halting probabilities}\label{MXWRrIRDRZ}
The \dce reals form a field under the usual addition and multiplication, as was
demonstrated by Ambos-Spies, Weihrauch, and Zheng 
\cite{Ambos.ea:00}. Raichev \cite{Raichev:05} and Ng \cite{Ng06} showed that this field is real-closed.
Downey, Wu and Zheng \cite{mlq/DowneyWZ04} studied the Turing degrees of \dce reals.
Clearly \dce reals are $\Delta^0_2$ since they can be computably approximated.
Downey, Wu and Zheng \cite{mlq/DowneyWZ04} showed that
every real which is truth-table reducible to the halting problem is Turing equivalent to
a \dce real. However they also showed that there are $\Delta^0_2$ degrees which
do not contain any \dce reals. In this strong sense, \dce reals form a strict subclass of the 
$\Delta^0_2$ reals.
\begin{defi}[Approximations]\label{1DDkA32C33}
A \lce approximation to a real $\alpha$ is a computable sequence of rationals $(\alpha_s)$
which converges to $\alpha$ and such that $\alpha_s\leq \alpha$ for almost all $s$.
A \dce approximation is a sequence of the form $(\alpha_s-\beta_s)$
where $(\alpha_s)$ and $(\beta_s)$ are bounded increasing sequences of rationals. 
\end{defi}
A real is called \rce if its {\em right} Dedekind cut is \ce or, equivalently, if
it is the limit of a decreasing computable sequence of rationals. A \rce approximation
is defined similarly to a \lce approximation but with the inequality reversed.
For example, a \dce approximation $(\alpha_s-\beta_s)$ to $\alpha-\beta$
is \lce if for almost all $s$ we have
$\alpha_s-\beta_s<\alpha-\beta$, and is \rce if for almost all $s$ we have
$\alpha_s-\beta_s>\alpha-\beta$.
A real is called {\em properly d.c.e.} if it is \dce and is neither a \lce nor a \rce real.
Rettinger and Zheng \cite{Rettinger2005} (also see \cite[Theorem 9.2.4 ]{rodenisbook}) 
proved that there are no
properly \dce \ml random reals.
\begin{equation}\label{eq:rett}
\parbox{13cm}{Every \ml random \dce real is a \lce or a \rce real. In fact, any \dce
approximation to a \ml random real is either \lce or \rce}
\end{equation}
An initial reaction to this result might be that there is not much to say about \dce reals and
\ml randomness. This impression quickly fades, however,  once rather basic questions are asked
regarding differences of \ml random \lce reals, which cannot be answered
by the existing theory of \ml random \lce reals. Problem \ref{3aCBk84Vgj} is such a question,
and is part of the more general question as to what properties the differences of \ml
random \lce reals have. The latter question becomes even more interesting once we recall
the maximality properties of \ml random \lce reals inside the class of \lce reals, which were
established largely on the basis of
the cumulative work of Solovay \cite{Solovay:75},
\CHKW \cite{Calude.Hertling.ea:01} and \KS \cite{Kucera.Slaman:01},
showing that the \ml random \lce reals are exactly the halting probabilities of
universal \pf machines. 
\DHN \cite{Downey02randomness} proved that given any two
\ml random \lce reals $\alpha,\beta$, if $q$ is a sufficiently large positive rational number then
$q\alpha-\beta$ is a \ml random \lce real. Similarly, if $q$ is
a sufficiently small positive rational number then
$q\alpha-\beta$ is a \ml random \rce real.
This result demonstrates the universality of \ml random \lce reals, in the sense that they can `absorb'
any other \lce real which is appropriately scaled. In other words, they remain \lce and \ml random
even when we subtract a \lce real from them, provided that the latter is appropriately scaled.

But what happens in-between these appropriate values of $q$? Is there a non-trivial interval of
rationals $q$ such that $q\alpha-\beta$ is neither a \lce nor a \rce real, and hence not a \ml random real?
Given two \lce reals $\alpha,\beta$ define:
%\begin{equation}\parbox{14cm}{
\begin{eqnarray}
\DD(\alpha,\beta)&=&\inf\Big\{\ q\in\QQ\ |\ \textrm{$q\alpha-\beta$ is a \lce real }\Big\}\\
\DD^{\ast}(\alpha,\beta)&=&\sup\Big\{\ q\in \QQ\ |\ \textrm{$q\alpha-\beta$ is a \rce real }\Big\}.
\end{eqnarray}
%}\end{equation}
As we discussed above, if $\alpha,\beta$ are \ml random then both
$\DD(\alpha,\beta)$ and $\DD^{\ast}(\alpha,\beta)$ are positive and finite. Moreover
$q\alpha-\beta$ is a \lce real for each rational  $q>\DD(\alpha,\beta)$ and 
a \rce real for each positive rational $q<\DD^{\ast}(\alpha,\beta)$. 

\begin{thm}\label{IXJbtDUFlf}
Given any \ml random \lce reals $\alpha,\beta$, we have $\DD(\alpha,\beta)=\DD^{\ast}(\alpha,\beta)$.
\end{thm}
This is remarkable! As $q$ ranges from $0$ to infinity, there is a unique point where
$q\alpha-\beta$ undergoes a sudden transition from being a \lce real to being a \rce real (note that this point of transition will not generally be $\beta/\alpha$ since $\DD(\alpha,\beta)$ and $\DD^{\ast}(\alpha,\beta)$ are unaffected by the addition or subtraction of rational values to $\alpha$ and $\beta$, so that $\DD(\alpha,\beta)=\DD(\alpha+\frac{1}{2},\beta)$, for example). Moreover, as will become
clear in the following, $q\alpha-\beta$ remains \ml random either side of the transition point, but loses its randomness precisely at the point of  transition. We can use this result in order to get a complete answer to the question
concerning differences of \ml random \lce reals.
\begin{thm}\label{S4IBp8JyfO}
Let $\alpha,\beta$ be \ml random \lce reals. Then
\begin{enumerate}[\hspace{0.5cm}(a)]
\item if $\DD(\alpha,\beta)<1$ then $\alpha-\beta$ is a \ml random  \lce real;
\item if $\DD(\alpha,\beta)=1$ then $\alpha-\beta$ is not \ml random;
\item if $\DD(\alpha,\beta)>1$ then $\alpha-\beta$ is a \ml random  \rce real.
\end{enumerate}
\end{thm}
The proof of Theorem \ref{IXJbtDUFlf} actually gives another remarkable new result regarding
the approximations to \ml random \lce reals. Given any two \ml random \lce reals $\alpha,\beta$
and any monotone \lce computable approximations $(\alpha_s)$, $(\beta_s)$ 
converging to $\alpha$ and $\beta$ respectively, the ratio of the rates of convergence has a unique limit, namely
$\DD(\alpha,\beta)>0$.
\begin{thm}\label{AONSVJ5jVl}
Let $\alpha,\beta$ be \ml random \lce reals and let $(\alpha_s)$, $(\beta_s)$ be any
monotone \lce computable approximations to $\alpha$ and $\beta$ respectively. Then:
\[
\lim_s \left(\frac{\alpha-\alpha_s}{\beta-\beta_s}\right)
\hspace{0.3cm}
\textrm{exists and is equal to $\DD(\alpha,\beta)$.}
\]
\end{thm}
In interpreting this theorem, note that $\DD(\alpha,\beta)$ depends only on $\alpha$ and $\beta$,  and is independent of the particular choice for
 $(\alpha_s)$, $(\beta_s)$.
Finally, we note that if $\alpha,\beta,\gamma$ are \ml random \lce reals and $q$ is a positive rational then: 
\begin{enumerate}[\hspace{0.5cm}(1)]
\item $\DD(\alpha+\beta,\gamma)= \DD(\alpha,\gamma) +\DD(\beta,\gamma)$
\item $\DD(\alpha,q\cdot\beta) =1/q\cdot \DD(\alpha,\beta)$
and $\DD(q\cdot\alpha,\beta) =q\cdot \DD(\alpha,\beta)$.
\item $\DD(\alpha,\beta)\cdot \DD(\beta,\alpha)= 1$. 
\end{enumerate}
These linearity properties of the operator $\DD$ are a direct consequence of
Theorem \ref{AONSVJ5jVl}. 

Miller \cite{derivationmiller} gives a beautiful account of the results
of this section, along with alternative proofs and extensions to the
\dce reals.

\subsection{Background}\label{subse:histcpro}
Some familiarity with the basic concepts of algorithmic information theory and the
basic methods of computability theory would be helpful for the reader. For such background
we refer to one of the monographs \cite{Li.Vitanyi:93, rodenisbook, Ottobook}, 
where the latter two are more focused on computability theory
aspects of algorithmic randomness.
In this section we lay out some preliminary material which is directly relevant to the present work
and which we avoided in the introduction for the sake of
clarity of presentation.  We start with a brief discussion of some facts from the theory of \lce reals. 
 The results we mention in Section \ref{5OrQ9kRZxr}
are directly relevant to our work, and will be used throughout our analysis. We continue
with some definitions of algorithmic 
randomness in Section \ref{9vkiTFyIef}. We define \ml randomness in two equivalent ways, namely
in terms of \ml tests and in terms of Solovay tests. Both of these notions will be used in our analysis; 
the argument of Section \ref{KjaAxAR79b} uses
a \ml test while Section \ref{fIxPfLMGKR} uses a Solovay test.
Finally in Section \ref{goTxrPK7Mz} we include a brief discussion of the original motivation behind
Problem \ref{3aCBk84Vgj}, which we omitted in the introduction.

\subsubsection{Completeness in the Solovay degrees of \lce reals}\label{5OrQ9kRZxr}
Solovay \cite{Solovay:75} defined a reducibility on the \lce reals as a measure of the hardness of
approximation.  Given \lce reals $\alpha,\beta$ we say that
$\alpha$ is Solovay reducible to $\beta$ if there exists a partial computable function $\varphi$ and
a constant $c$ such that
for all rationals $q<\beta$, $\varphi(q)$ is defined and $0\leq \alpha-\varphi(q)\leq c\cdot (\beta-q)$.
This condition is equivalent to requiring that there are \lce monotone approximations 
$(\alpha_s), (\beta_s)$ to $\alpha,\beta$ respectively and
a constant $c$ such that $\alpha-\alpha_s\leq c\cdot (\beta-\beta_s)$
for all $s$. \DHN \cite{Downey02randomness} showed that $\alpha$ is Solovay reducible
to $\beta$ if and only if there exists a \lce real $\gamma$ and a rational constant $c$ such that 
$\beta=c\cdot\alpha+\gamma$.
In the same paper it was shown that 
$\alpha$ is Solovay reducible to $\beta$ if and only if there are 
 \lce monotone approximations 
$(\alpha_s), (\beta_s)$ to $\alpha,\beta$ respectively and a constant $c$ such that
$\alpha_{s+1}-\alpha_s< c\cdot (\beta_{s+1}-\beta_s)$ for all $s$.
Solovay's work, combined with the work of
\CHKW \cite{Calude.Hertling.ea:01} and \KS \cite{Kucera.Slaman:01}, showed that
the \ml random \lce reals are exactly the complete \lce reals with respect to Solovay reducibility.
Moreover, this work also implies that the  \ml random \lce reals are exactly the halting probabilities of
universal \pf machines. 

More specifically, \CHKW \cite{Calude.Hertling.ea:01} showed that if 
$\alpha,\beta$ are \lce reals, $\alpha$ is \ml random and
is Solovay reducible to $\beta$ then $\beta$ is \ml random.
\KS showed that every \ml random \lce real is complete with respect to Solovay reducibility.
\DHN \cite{Downey02randomness} also showed that addition is a join operation in the Solovay degrees
and that the \ml random \lce reals are join-irreducible. In other words, if $\alpha$ is a 
\ml random \lce real and $\beta$ is another \lce real then $\alpha+\beta$ is a \ml random \lce real;
conversely, if $\alpha,\beta$ are \lce reals and  $\alpha+\beta$ is \ml random, then at least one of 
$\alpha,\beta$ is \ml random. The first of these results had already been proved in 
Demuth \cite{Dempseu}.

The starting point of the present work was to consider the size of the constant $c$ in Solovay reductions.
In particular, any two \ml random \lce reals are Solovay equivalent, but how small can the constants
in the associated reductions be? If one examines the proofs in the work mentioned above, one finds
 that the infimum of the constants $c$ that are possible in a reduction between two 
\lce reals $\alpha,\beta$, is the same for all of the above 
characterisations of Solovay reducibility. This observation
indicated that this infimum is a robust characteristic of the pair $\alpha,\beta$ --  a fact that
we eventually proved with Theorem \ref{IXJbtDUFlf} and Theorem \ref{AONSVJ5jVl}.

\subsubsection{\ml and Solovay tests for randomness}\label{9vkiTFyIef}
Algorithmic randomness for real numbers 
was originally defined by \ml in \cite{MR0223179} in terms of effective statistical tests.
A \ml test is a uniformly \ce sequence of $\Sigma^0_1$ classes of reals $(U_i)$ such that
the Lebesgue measure of $U_i$ is bounded by $2^{-i}$. In other words, a \ml test is a uniform 
sequence of effectively open sets of reals whose measure decreases uniformly. 
Effectively open sets of reals are often represented as \ce sets of finite strings, while finite strings in turn can be viewed
as basic open sets in the Cantor space (a string $\sigma$ represents the set of reals with prefix $\sigma$).
A point on the real line can be regarded as  a
member of the Cantor space (consisting of the infinite binary sequences) by identifying it with its
binary expansion. 
We say that
a real $\alpha$ is \ml random if $\alpha\not\in \cap_i U_i$ for all \ml tests $(U_i)$.

Solovay \cite{Solovay:75} devised an equivalent way to test \ml randomness. 
A Solovay test is a uniformly \ce sequence of  $\Sigma^0_1$ classes $I_j$ such that
the sum of the measures $\mu(I_j)$ (summing over all $j\in\Nat$) is finite.
Solovay showed that a real $\alpha$ is \ml random if and only if, for every Solovay test $(I_j)$
there are only finitely many $j$ with  $\alpha\in I_j$.

\subsubsection{Historical context of the Problem \ref{3aCBk84Vgj}}\label{goTxrPK7Mz}
Some researchers have been looking to obtain other algorithmically random numbers, perhaps
more random than the halting probability, through an examination of the stochastic behavior of 
a Turing machine. Here the stochasticity refers to some form of probability measure on the space of  inputs, which can be seen as
the outcomes of an experiment, like a repeated coin-toss. The strength of algorithmic randomness
can be calibrated through the various recursion-theoretic hierarchies. The standard notion of algorithmic
randomness is  \ml randomness, also called 1-randomness, which can be relativized to the $n$th
iteration of the halting problem, giving $n$-randomness.
Becher and Chaitin \cite{firstBC} 
considered a prefix-free model for infinite computations,
which was introduced by Chaitin in \cite{Chaitin1976233} and whose study was
suggested in \cite[Chapter 6]{chaitin2001exploring}. 
They showed that the probability that a universal prefix-free machine for
infinite computations outputs finitely many symbols is 2-random.

According to \cite{jsyml/BecherFGM06}, in 2002 Grigorieff suggested 
the study of the halting probability of a prefix-free machine restricted to
a set $X$ of outputs. He considered
the number \eqref{eq:UhaltproX},
which is the probability that a standard universal prefix-free machine $U$
halts with output a member of the set $X$. 
He conjectured that the {\em higher the arithmetical complexity of $X$ is, the more random the
number $\Omega_U(X)$ becomes.}
If this were true, then this would be a neat way to obtain highly random numbers as halting probabilities
in the standard prefix-free model for computation.
Versions of of this conjecture were successfully pursued for Chaitin's infinite computation model
of   \cite{Chaitin1976233}, by Becher and Grigorieff in
\cite{DBLP:journals/jsyml/BecherG05, jsyml/BecherG09}. However,
as far as the usual prefix-free machines are concerned, the conjecture proved largely false
after Miller in 2004 produced
a $\Delta^0_2$ set $X$ such that 
$\Omega_U(X)$ is a rational number 
(see \cite{DBLP:journals/jsyml/BecherG05} for the announcement of
this result and \cite{jsyml/BecherFGM06} for a proof).
\BFGM \cite{jsyml/BecherFGM06} continued with a large number of negative results, showing that
$\Omega_U(X)$ is often less random than expected, given the complexity of $X$. In this same paper
there is also a positive result, showing that $\Omega_U(X)$ is 1-random when $X$ is
$\Sigma^0_n$-complete or $\Pi^0_n$-complete for some $n>1$. Nevertheless, they also show that
for each $n>1$ the number 
$\Omega_U(X)$ 
is not $n$-random for any 
$\Sigma^0_n$ or $\Pi^0_n$ set $X$.
On a positive note, Chaitin \cite{chaitin2004algorithmic} had already noticed that if $X$ is a $\Sigma^0_1$ set, 
then $\Omega_U(X)$ is \ml random. The case when $X$ is \pz is the basis of 
Problem \ref{3aCBk84Vgj} and has
remained open.
\FSW
\cite{jc/FigueiraSW06} constructed an optimal (in terms of Kolmogorov complexity) 
but not universal \pf machine $U$ and a \pz set $X$ such that  $\Omega_U(X)$ is not \ml random. 
They also suggested some strategies for giving a negative answer to Problem \ref{3aCBk84Vgj}.
In the present work, we see that Problem \ref{3aCBk84Vgj} has a surprising positive answer.

\section{Overview}
In this section we reduce all of the results in this paper to two technical lemmas, which
we prove in Sections \ref{fIxPfLMGKR} and \ref{KjaAxAR79b} respectively. This style of presentation
should make the content of this work readily accessible. 
In Section \ref{IIQvorJh4q} we derive the results of Section \ref{MXWRrIRDRZ} from
Lemma \ref{seclem}, whose proof is a delicate 
technical argument which is defferred to Section \ref{fIxPfLMGKR}.
Then in Section \ref{kJtnt1Dzg6} we give the solution of Problem 
\ref{3aCBk84Vgj}, using the results of Section \ref{MXWRrIRDRZ} and Lemma 
\ref{wd9LMyiLha}. The proof latter lemma is a {\em decanter argument}, 
which has some things in common
with the original decanter argument by \DHNS \cite{Downey.Hirschfeldt.ea:03}
and is given in Section  \ref{KjaAxAR79b}.

\subsection{Proof of Theorem  \ref{IXJbtDUFlf} and Theorem  \ref{AONSVJ5jVl}}\label{IIQvorJh4q}
\DHN \cite{Downey02randomness} showed that, given any two \lce reals $\alpha,\beta$,
the number $\DD(\alpha,\beta)$ is the limit infimum of $(\alpha-\alpha_s)/(\beta-\beta_s)$
with respect to all \lce increasing approximations $(\alpha_s)$, $(\beta_s)$ to $\alpha,\beta$ respectively.
Hence for Theorem \ref{AONSVJ5jVl} it suffices to show that in the case where $\alpha,\beta$
are \ml random, given
\lce increasing approximations $(\alpha_s)$, $(\beta_s)$ to $\alpha,\beta$ respectively,
 the limit of $(\alpha-\alpha_s)/(\beta-\beta_s)$ exists.
By the density of the rationals in the real line, this statement follows from the following
special case, which is proved in
Section \ref{fIxPfLMGKR}.
\begin{lem}  \label{seclem} 
Suppose $\alpha$ and $\beta$ are c.e.\ reals, that $\alpha$ is Martin-L\"{o}f random and that  $p\in \mathbb{Q}$ with $p>1$. Suppose given \lce approximations $(\alpha_s)$ and $(\beta_s)$ to $\alpha$ and $\beta$ respectively. Then the following two conditions cannot both hold: 
\begin{enumerate} 
\item There exist infinitely many $s$ with $p(\alpha -\alpha_s)<(\beta-\beta_s)$;
\item There exist infinitely many $s$ with $ (\beta-\beta_s)< (\alpha -\alpha_s)$. 
\end{enumerate}
\end{lem} 
Next, we show how to derive Theorem \ref{IXJbtDUFlf} from Lemma \ref{seclem}.
Suppose that $\alpha$ and $\beta$ are c.e.\ reals and that $\alpha$ is Martin-L\"{o}f random. In order to establish Theorem \ref{IXJbtDUFlf}, it suffices to show  that if $p,q\in \mathbb{Q}^+$ with $p>q$, then either $p\alpha-\beta$ is \lce or else $q\alpha -\beta$ is a \rce real.\footnote{Indeed, if
Theorem \ref{IXJbtDUFlf} did not hold, then
$\DD(\alpha,\beta)>\DD^{\ast}(\alpha,\beta)$. If we take $q<p$ between
$\DD^{\ast}(\alpha,\beta)$ and $\DD(\alpha,\beta)$ then we would have that
$p\alpha-\beta$ is not \lce and $q\alpha -\beta$ is not a \rce real.}  
In fact it suffices to establish this for the case $q=1$. 
Then, for general $p>q$ in $\mathbb{Q}^+$, if $q\alpha -\beta$ 
is not right c.e., neither is $\alpha-(\beta/q) $, which means that 
$(p/q)\alpha-(\beta/q)$ and thus $p\alpha -\beta$ must be a \lce real.  So to prove 
\ref{IXJbtDUFlf}, it suffices to establish the following fact. 
\begin{equation}\label{mainlem}
\parbox{12cm}{If $\alpha$ and $\beta$ are c.e.\ reals and $\alpha$ is Martin-L\"{o}f random, then for $p\in \mathbb{Q}$ with $p>1$, either $p\alpha-\beta$ is left-c.e., or else $\alpha-\beta$ is a \rce real.}
\end{equation}

Now we show how to obtain \eqref{mainlem} from Lemma \ref{seclem}. 
Consider \lce approximations $(\alpha_s)$ and $(\beta_s)$ to $\alpha$ and $\beta$ such that $\alpha$ is Martin-L\"{o}f random. There are then two cases to consider. 

{\em Case 1.} There do not exist infinitely many $s$ with $p(\alpha -\alpha_s)<(\beta-\beta_s)$. 
If there exists an $s$ with $p(\alpha -\alpha_s)=(\beta-\beta_s)$ then $p\alpha-\beta$ is rational and so is clearly a \lce real. Otherwise, we can find $s_0$ such that  $p(\alpha -\alpha_s)>(\beta-\beta_s)$ for all $s\geq s_0$. In order to define a \lce approximation $(\gamma_s)$ to $p\alpha-\beta$, define $\gamma_0= p\alpha_{s_0}-\beta_{s_0}$. Given $\gamma_i$ and $s_i$ (such that $p(\alpha -\alpha_{s_i})>(\beta-\beta_{s_i})$), find $s_{i+1}>s_i$ such that $ p(\alpha_{s_{i+1}}-\alpha_{s_i})>(\beta_{s_{i+1}}- \beta_{s_i})$, and define $\gamma_{i+1}= p\alpha_{s_{i+1}}-\beta_{s_{i+1}}$.

{\em Case 2.} There do exist infinitely many $s$ with $p(\alpha -\alpha_s)<(\beta-\beta_s)$. In this case Lemma \ref{seclem} tells us that there cannot exist infinitely many $s$ with $ (\beta-\beta_s)< (\alpha -\alpha_s)$. If there exists an $s$ with $ (\beta-\beta_s)=(\alpha -\alpha_s)$ then $\alpha-\beta$ is rational and so is a \rce real. Otherwise, let $s_0$ be such that  $(\alpha -\alpha_s)<(\beta-\beta_s)$ for all $s\geq s_0$. In order to define a right c.e. approximation $(\delta_s)$ to $\alpha-\beta$, define $\delta_0= \alpha_{s_0}-\beta_{s_0}$. Given $\delta_i$ and $s_i$ (such that $(\alpha -\alpha_{s_i})<(\beta-\beta_{s_i})$), find $s_{i+1}>s_i$ such that $( \alpha_{s_{i+1}}-\alpha_{s_i})<(\beta_{s_{i+1}}- \beta_{s_i})$, and define $\delta_{i+1}= \alpha_{s_{i+1}}-\beta_{s_{i+1}}$.

We have reduced Theorem  \ref{IXJbtDUFlf} and Theorem  \ref{AONSVJ5jVl} to
Lemma  \ref{seclem}, whose proof is given in Section \ref{fIxPfLMGKR}.

\subsection{Proof of Theorem \ref{S4IBp8JyfO}}
For clause (a) suppose that $\DD(\alpha,\beta)<1$. Then there exists
a rational number $q<1$ such that $q\alpha-\beta$ is a \lce real.
Recall that Demuth \cite{Dempseu} showed that the sum of a \ml random \lce real and
any other \lce real is \ml random. Therefore $\alpha-\beta$ is \ml random as the sum of
the \ml random \lce real $(1-q)\alpha$ and the \lce real $q\alpha-\beta$.

For clause (c) the argument is similar. If $\DD(\alpha,\beta)>1$ 
then by Theorem \ref{IXJbtDUFlf} we have
$\DD^{\ast}(\alpha,\beta)>1$. 
So there exists $p>1$ such that $p\alpha-\beta$ is a \rce real.
Recall that Demuth \cite{Dempseu} showed that the sum of a \ml random \rce real and
any other \rce real is \ml random. 
Hence $\alpha-\beta$ is a \ml random \rce real as the sum of the \ml random 
\rce real $(1-p)\alpha$ and the \rce real $p\alpha-\beta$.

For clause (b),  it suffices to prove the contrapositive. So assume that
$\alpha-\beta$ is \ml random. Then by \eqref{eq:rett}, 
$\alpha-\beta$ is either a \lce real or a \rce real.
Without loss of generality, assume that it is a \ml random \lce real. Then by
\DHN \cite{Downey02randomness} there exists a sufficiently small but positive rational $q$ such that
$(\alpha-\beta) -q\alpha$ is a \ml random \lce real. But this means that $\DD(\alpha,\beta)\leq 1-q$
so $\DD(\alpha,\beta)<1$.
In the case where
$\alpha-\beta$ 
is a \ml random \rce real, a similar argument gives that
$\DD(\alpha,\beta)>1$. In any case, $\DD(\alpha,\beta)\neq 1$, which concludes the proof.

\subsection{Proof of Theorem \ref{JgIfqiIPd}}\label{kJtnt1Dzg6}
Let us briefly note that \eqref{A69imMhIZ}
is a direct consequence of Theorem \ref{JgIfqiIPd} and \eqref{eq:rett}.\footnote{If \eqref{A69imMhIZ}
did not hold, then either there are infinitely many terms of the sequence $(\alpha_s)$ on both sides
of $\Omega_U(X)$, or almost all terms are larger than $\Omega_U(X)$. In the first case, by 
\eqref{eq:rett} we get that $\Omega_U(X)$ is not \ml random, which contradicts Theorem \ref{JgIfqiIPd}.
In the second case we get that $\Omega_U(X)$ is \rce and since by Theorem \ref{JgIfqiIPd} it is also 
\lce it has to be computable. But then again this contradicts Theorem \ref{JgIfqiIPd}.}
We are going to reduce Theorem \ref{JgIfqiIPd} to the following two lemmas, of which
the first one has a self-contained proof presented in  Section \ref{KjaAxAR79b}, and the latter one
is based on the results of Section \ref{MXWRrIRDRZ}.

\begin{lem}\label{wd9LMyiLha}
If $X$ is a \pz set and \omux is a \rce real then
 \omux is not \ml random.
\end{lem}
The following lemma can be proved using Theorem \ref{S4IBp8JyfO}, and is
of independent interest.

\begin{lem}\label{jzaprnNjHH}
If $\delta$ is a \dce real which is not \ml random and $\alpha$ is a \ml random \lce real then
$\alpha+\delta$ is a \ml random \lce real.
\end{lem}
\begin{proof}
Let $\delta=\delta_0-\delta_1$. We can suppose that $\delta_0$ and $\delta_1$ are \ml random (otherwise add a random \lce real to both). Since $\delta$ is not \ml random we have $\DD(\delta_0,\delta_1)=1$, 
by Theorem \ref{S4IBp8JyfO}.
Since $\alpha$ is a \ml random \lce real, by  \DHN \cite{Downey02randomness}
(see the discussion in Section \ref{5OrQ9kRZxr})
there exists some $c$ and another \lce real $\beta$ such that 
$\alpha=2^{-c}\cdot \delta_0+\beta$. So 
$\alpha+\delta=\beta+(1+2^{-c})\cdot \delta_0-\delta_1$.
But since  $\DD(\delta_0,\delta_1)=1$ the real $(1+2^{-c})\cdot \delta_0-\delta_1$
is \lce and \ml random (again by \cite{Downey02randomness}). So $\alpha+\delta$ is also  \lce and \ml random, which completes the proof.
\end{proof}
We are now ready to prove Theorem \ref{JgIfqiIPd}, using Lemma
\ref{wd9LMyiLha} and Lemma \ref{jzaprnNjHH}.
Let $x\in X$ and consider $Y=X-\{x\}$ which is a \pz set. Then 
$\Omega_U(X)=\Omega_U(Y)+\Omega_U(\{x\})$.
Note that $\Omega_U(\{x\})$ is a \ml random \lce real.
By Lemma \ref{wd9LMyiLha} the number
$\Omega_U(Y)$ is one of the following:
\begin{enumerate}[\hspace{0.5cm}(a)]
\item a \rce real which is not \ml random;
\item a proper \dce real (hence, by \eqref{eq:rett}, not \ml random);
\item a \lce real.
\end{enumerate}
In the first two cases $\Omega_U(Y)$ is a \dce real which is not \ml random, so 
by Lemma \ref{jzaprnNjHH} and the fact that $\Omega_U(\{x\})$ is a \ml random \lce real
we have that $\Omega_U(X)$ is a \ml random \lce real. In  the last case
 $\Omega_U(Y)$ is a \lce real so $\Omega_U(X)$ is again a \ml random \lce real as the sum of
 a \lce real and a random \lce real, which concludes the proof.

\section{Proof of Lemma \ref{seclem}}\label{fIxPfLMGKR}
Towards a contradiction, suppose that $\alpha,\beta,p,(\alpha_s)$ and $(\beta_s)$ are as in the statement of the Lemma \ref{seclem}, and that conditions (1) and (2) both hold. We shall show how to enumerate a Solovay test establishing that $(p-1)\alpha$ is not random. 
It will be convenient to make use of the following notation. 

\begin{defi} 
 For any given $s>0$, let $\alpha^{\ast}_s=\alpha_s-\alpha_{s-1}$. For $s_1\geq s_0$ define $\alpha^{\ast}_{s_0,s_1}=\alpha_{s_1}-\alpha_{s_0}$, and  $\alpha^{\ast}_{s_0,\infty}=\alpha -\alpha_{s_0}$. Also define $\beta^{\ast}_s$, $\beta^{\ast}_{s_0,s_1}$ and $\beta^{\ast}_{s_0,\infty}$ similarly.
\end{defi}

From the fact that condition (1) holds, we can actually deduce that there exist infinitely many $s$ for which a stronger condition holds: 

\begin{enumerate} 
\item[$(\ast 1)$] For all $s'>s$, $p\alpha^{\ast}_{s,s'}<\beta^{\ast}_{s,s'}$. 
\end{enumerate}
In order to see this, let $s_0$ be such that $p\alpha^{\ast}_{s_0,\infty}<\beta^{\ast}_{s_0,\infty}$. 
If  $(\ast 1)$ holds for $s_0$ then we have found
a stage $\geq s_0$ which satisfies $(\ast 1)$. Otherwise,
let $s\geq s_0$ be minimal such that, for all $s'>s$, 
$p\alpha^{\ast}_{s_0,s'}<\beta^{\ast}_{s_0,s'}$. This means that 
$p\alpha^{\ast}_{s_0,s}\geq \beta^{\ast}_{s_0,s}$. 
If $s'>s$ and $p\alpha^{\ast}_{s,s'}\geq \beta^{\ast}_{s,s'}$, 
then $p\alpha^{\ast}_{s_0,s'}=p(\alpha^{\ast}_{s_0,s}+\alpha^{\ast}_{s,s'})
\geq \beta^{\ast}_{s_0,s}+\beta^{\ast}_{s,s'}=\beta^{\ast}_{s_0,s'}$, 
a contradiction. So $s$ is a stage $\geq s_0$ which satisfies $(\ast 1)$. 
In each case, we  can conclude that there exists
a stage $\geq s_0$ which satisfies $(\ast 1)$.

From (2) we can similarly deduce that there are infinitely many $s$ for which the following stronger condition holds: 

\begin{enumerate} 
\item[$(\ast 2)$] For all $s'>s$, $\beta^{\ast}_{s,s'}<\alpha^{\ast}_{s,s'}$. 
\end{enumerate}

In order to see this, let $s_0$ be such that 
$\beta^{\ast}_{s_0,\infty}<\alpha^{\ast}_{s_0,\infty}$. 
If  $(\ast 2)$ holds for $s_0$ then we have found
a stage $\geq s_0$ which satisfies $(\ast 2)$. Otherwise,
let $s\geq s_0$ be minimal such that, for all 
$s'>s$, $\beta^{\ast}_{s_0,s'}<\alpha^{\ast}_{s_0,s'}$. This means that 
$\beta^{\ast}_{s_0,s} \geq \alpha^{\ast}_{s_0,s}$. 
If $s'>s$ and $ \beta^{\ast}_{s,s'}\geq \alpha^{\ast}_{s,s'}$, 
then $\alpha^{\ast}_{s_0,s'}=\alpha^{\ast}_{s_0,s}+
\alpha^{\ast}_{s,s'}\leq \beta^{\ast}_{s_0,s}+
\beta^{\ast}_{s,s'}=\beta^{\ast}_{s_0,s'}$, a contradiction. 
So $s$ is a stage $\geq s_0$ which satisfies $(\ast 2)$. 
In each case, we  can conclude that there exists
a stage $\geq s_0$ which satisfies $(\ast 2)$.

 Without loss of generality we may suppose that $s=0$ is one of those stages for which $(\ast 1)$ holds 
(otherwise begin the enumeration at some later stage). 

\subsection{Proof idea}

The key to the proof lies in the following basic observation.
Suppose that $s_1>0$ is the least stage for which $(\ast 2)$ holds.  Let $\beta^1=\beta^{\ast}_{0,s_1}$, $\beta^2=\beta^{\ast}_{s_1,\infty}$ and $\alpha^1=\alpha^{\ast}_{0,s_1}$, $\alpha^2=\alpha^{\ast}_{s_1,\infty}$ (where 1 and 2 are used as indices rather than to denote exponentiation). Then we have: 
\begin{equation} 
\beta^1+\beta^2>p(\alpha^1 +\alpha^2), \ \ \ \ \beta^2<\alpha^2,
\end{equation}
which gives,
\begin{equation}  \label{eq1}
\beta^1-\alpha^1>(p-1)(\alpha^1+\alpha^2). 
\end{equation}

In order to begin enumerating a Solovay test which will be failed by $(p-1)\alpha$, (\ref{eq1}) suggests that we could start by enumerating intervals in the following fashion. Let $a_1=(p-1)\alpha_0$ and $b_1=a_1+(\beta^{\ast}_1-\alpha^{\ast}_1)$. At stage 1 we enumerate the interval $[a_1,b_1]$. At later stages $s$, we could consider the last interval $[a_i,b_i]$ enumerated, and look to see whether $a_1+(\beta^{\ast}_{0,s}-\alpha^{\ast}_{0,s})>b_i$. If so, then we could define $a_{i+1}=b_i$, $b_{i+1}=a_1+(\beta^{\ast}_{0,s}-\alpha^{\ast}_{0,s})$ and enumerate this into our Solovay test. What happens if we proceed in this way? The first thing to observe is that, by the end of stage $s_1$ (recall that $(\ast 2)$ holds for $s_1$),  (\ref{eq1}) ensures $(p-1)\alpha$ meets at least one element of the test enumerated thus far.  What then happens after stage $s_1$? The fact that $s_1$ satisfies $(\ast 2)$ means that at every stage $s>s_1$ we will find that $a_1+(\beta^{\ast}_{0,s}-\alpha^{\ast}_{0,s})=a_1+(\beta^{\ast}_{0,s_1}-\alpha^{\ast}_{0,s_1})+(\beta^{\ast}_{s_1,s}-\alpha^{\ast}_{s_1,s})<a_1+ (\beta^{\ast}_{0,s_1}-\alpha^{\ast}_{0,s_1})$. So the rules described so far mean that we will not enumerate any further elements into our Solovay test attempting to capture our first initial segment of $(p-1)\alpha$ after stage $s_1$.  This is good news, because to do so would be wasteful -- we have already captured one initial segment of $(p-1)\alpha$, and subsequent enumerations into the test should be aimed at capturing another initial segment. 

During the construction of our Solovay test we shall place movable markers $x_0<y_0<x_1<y_1<\dots $ on stages, which are our present guesses as to stages which satisfy the conditions $(\ast 1)$ and $(\ast 2)$.  We shall have $x_0=0$, since we have assumed that $s=0$ satisfies $(\ast 1)$. If $y_0\downarrow (>0)$ then our present guess is that this is a stage which satisfies $(\ast 2)$, while $x_1$ may be placed on a stage which we believe satisfies $(\ast 1)$ again, and so on.  Suppose that at some stage $s$, we have $y_0=s_1$ and $x_1\downarrow >y_0$. Then at stages $s'>x_1$ we shall enumerate elements into our Solovay test of the form $[a_1',b_1'],[a_2',b_2'],\dots$ but now beginning with $a_1'=(p-1)\alpha_{x_1}$. Now for the sake of argument, suppose that $s_1$ \emph{does not} in fact satisfy $(\ast 2)$ -- that it merely appeared to be a stage satisfying this property, up until some later stage $s>x_1$. At this point it might initially seem that we have a difficulty, because the measure we have enumerated into the sequence $[a_1',b_1'],[a_2',b_2'],\dots$ should actually be measure we enumerate into the first sequence $[a_1,b_1],[a_2,b_2],\dots$ attempting to capture our first initial segment of $(p-1)\alpha$. Thus it may seem that we are threatened with having to enumerate certain measure into our Solovay test twice.  The fact that $x_1$ appeared to satisfy $(\ast 2)$ up until stage $s$, however, means that this is not a problem. We have $\beta^{\ast}_{s_1,s-1}-\alpha^{\ast}_{s_1,s-1}<0$, while $\beta^{\ast}_{s_1,s}-\alpha^{\ast}_{s_1,s}\geq 0$. One way to look at this is to see that, for some $r$ with $0<r\leq 1$ we have  $(\beta^{\ast}_{s_1,s-1}-\alpha^{\ast}_{s_1,s-1}) + r(\beta^{\ast}_s-\alpha^{\ast}_s)=0$, and that enumerating the measure $\beta^{\ast}_{s_1,s}-\alpha^{\ast}_{s_1,s}$ into our Solovay test actually only requires us to enumerate the new measure $(1-r)(\beta^{\ast}_s-\alpha^{\ast}_s)$. 

At stage $s+1$ of the construction, we shall say that $x_i$ \emph{requires moving} if $x_i\downarrow $ and $\beta^{\ast}_{x_i,s+1}\leq p\alpha^{\ast}_{x_i,s+1}$. We say that $y_i$ requires moving if $y_i\downarrow $ and $\beta^{\ast}_{y_i,s+1}\geq \alpha^{\ast}_{y_i,s+1}$. When we enumerate an interval into the Solovay test, it will always be the case that we associate that element of the Solovay test with one of the markers $x_i$. If the marker $x_i$ is subsequently made undefined (because it or a marker of higher priority requires moving), then we cannot change the fact that elements of the Solovay test already associated with $x_i$ have been enumerated, but we no longer consider them to be associated with $x_i$, i.e.\ the list of intervals associated with $x_i$ is initialised. 

\subsection{Construction} 
At stage 0, place the marker $x_0$ on 0. 

At stage $s+1$, let $z$ be the least number on which a marker is placed and which requires moving, or if there exists no such number, then leave $z$ undefined. There are three cases to consider. 

\begin{enumerate}[\hspace{0.5cm}(1)]
\item $z\downarrow =x_i$ for some $i$. In this case, make $x_{i'}$ and $y_{i'}$ undefined for all $i'\geq i$, and go to the next stage. 
\item $z\downarrow =y_i$ for some $i$.  Let  $r$ be such that  $(\beta^{\ast}_{y_i,s}-\alpha^{\ast}_{y_i,s}) + r(\beta^{\ast}_{s+1}-\alpha^{\ast}_{s+1})=0$. Let $[a,b]$ be the most recently enumerated interval associated with $x_i$. Enumerate the interval $[b,b+(1-r)(\beta^{\ast}_{s+1}-\alpha^{\ast}_{s+1})]$ into the Solovay test, and associate it with $x_i$. Make $y_i$ undefined, and make all $x_{i'}$ and $y_{i'}$ undefined for $i'>i$. Go to the next stage. 
\item $z\uparrow$. In this case, let $w$ be the greatest number on which a marker is placed. There are then two subcases. 
\begin{enumerate}[(a)]
\item  $w=x_i$ for some $i$. If $\beta^{\ast}_{s+1}\geq \alpha^{\ast}_{s+1}$, then let $[a,b]$ be the most recently enumerated interval associated with $x_i$, or let $b=(p-1)\alpha_0$ if there exists no such, and enumerate the interval $[b, b+(\beta^{\ast}_{s+1} -\alpha^{\ast}_{s+1})]$ into the Solovay test, associating it with $x_i$. If $\beta^{\ast}_{s+1}<\alpha^{\ast}_{s+1}$ then place the marker $y_i$ on $s$. In either case,  go to the next stage. 
\item $w=y_i$ for some $i$. In this case, if $\beta^{\ast}_{s+1}>p \alpha^{\ast}_{s+1}$ proceed as follows, and otherwise proceed immediately to the next stage. Place the marker $x_{i+1}$ on $s$. Enumerate the interval 
\[
[(p-1)\alpha_s, (p-1)\alpha_s + (\beta^{\ast}_{s+1} - \alpha^{\ast}_{s+1})]
\]
into the Solovay test, and associate it with $x_{i+1}$. Go to the next stage. 
\end{enumerate}

\end{enumerate}

\subsection{Verification} 
Since the total measure enumerated into the Solovay test is bounded by $\beta$, it must be finite. It is also clear that if a marker is placed at some stage $s$ and never subsequently moved, then $s$ must satisfy $(\ast 1)$ if the marker is $x_i$ for some $i$, and $s$ must satisfy $(\ast 2)$ if the marker is $y_i$ for some $i$.

 In order to see that each marker is eventually permanently placed, suppose inductively that $x_0,y_0,\dots, x_i,y_i$ are all correctly placed at some stage $y_i +1$ (meaning that the marker $y_i$ has just been placed correctly). Let $s>y_i$ be the least such that $(\ast 1)$ holds for $s$, and consider what happens at stage $s+1$. If the marker $x_{i+1}$ is already placed on a number $<s$, then it must be the case that $\beta^{\ast}_{x_{i+1},s}>p\alpha^{\ast}_{x_{i+1},s}$, otherwise the marker $x_{i+1}$ would have been moved prior to stage $s+1$. But then, for any $s'>s$,  $\beta^{\ast}_{x_{i+1},s'}= \beta^{\ast}_{x_{i+1},s}+\beta^{\ast}_{s,s'}>p(\alpha^{\ast}_{x_{i+1},s}+\alpha^{\ast}_{s,s'})= p\alpha^{\ast}_{x_{i+1},s'}$, and so $x_{i+1}$ satisfies $(\ast 1)$, contradicting the minimality of $s$. Thus $x_{i+1}$ cannot be placed at the start of stage $s+1$, and will be placed on $s$ at this stage. 
 
 Now suppose that $x_0,y_0,\dots,x_{i}$ have been placed correctly, let $s$ be the least $>x_{i}$ satisfying $(\ast 2)$, and consider what happens at stage $s+1$.  If the marker $y_{i}$ is already placed on a number $<s$, then it must be the case that $\beta^{\ast}_{y_i,s}<\alpha^{\ast}_{y_i,s}$, otherwise the marker $y_i$ moved have been moved prior to stage $s+1$. But then, for any $s'>s$,  $\beta^{\ast}_{y_i,s'}= \beta^{\ast}_{y_i,s}+\beta^{\ast}_{s,s'}<\alpha^{\ast}_{y_i,s}+\alpha^{\ast}_{s,s'}= \alpha^{\ast}_{x_i,s'}$, and so $y_i$ satisfies $(\ast 2)$, contradicting the minimality of $s$. Thus $y_i$ cannot be placed at the start of stage $s+1$, and will be placed on $s$ at this stage.
 
 Finally we note that once each $x_i$ is correctly placed, an interval to which $(p-1)\alpha$ belongs is subsequently enumerated into the Solovay test and associated with $x_i$. This follows because, letting $y_i$ take its final value, intervals are enumerated into the test and associated with $x_i$ covering the entire interval $[(p-1)\alpha_{x_i},   (p-1)\alpha_{x_i} +(\beta^{\ast}_{x_i,y_i}-\alpha^{\ast}_{x_i,y_i})]$. It follows from (\ref{eq1}) that $(p-1)\alpha$ belongs to this interval.  

\section{Proof of Lemma \ref{wd9LMyiLha}}\label{KjaAxAR79b}
Let $X$ be a \pz set. Given the existence of effective bijections between $\Nat$ and $2^{<\omega}$, it does not matter whether we consider $X$ as a subset of $\Nat$ or $2^{<\omega}$, and so for convenience we shall think of $X$ as a set of natural numbers. We must show that if \omux is \rce then it cannot be \ml random.   If $\Nat-X$ is finite then
\omux is left-c.e., so if it was \rce then it would be computable and so not \ml random. 
We may therefore assume that $\Nat-X$ is infinite.  If \omux  is rational then it is clearly not \ml random, as required.
So we may assume that \omux  is not rational.
Let $(X_s)$ be a \pz approximation to $X$ and consider the sequence of rationals
$(\Omega_U(X)[s])$ which converges to \omuxn.
If all but finitely many terms of  $(\Omega_U(X)[s])$ lie to the left of
the limit \omux then  \omux  is a \lce real, so if \rce it would be computable and so not \ml random.  If there are infinitely many terms of
$(\Omega_U(X)[s])$ on either side of the limit \omux then
\omux is not \ml random by \eqref{eq:rett}, as required.
So it remains to consider the case where 
all but finitely many terms of  $(\Omega_U(X)[s])$ lie to the right of
the limit \omuxn. In this case, since \omux is not rational, there exists an increasing 
sequence of stages $(s_i)$
such that $\Omega_U(X)[s_{i+1}]<\Omega_U(X)[s_{i}]$ for all $i$.
In particular, the terms $\Omega_U(X)[s_{i}]$ form a decreasing sequence which converges
to \omuxn. Define $\omega_i=\Omega_U(X)[s_{i}]$ for each $i$. It remains to show that in this case
 \omux is not \ml random. We will do this by constructing
a \ml test $(U_i)$ such that $\omuxf\in U_i$ for all $i$.

%\paragraph{\bf Description of the decanter.}
\subsection{Description of the decanter}
At each stage $s_{i+1}$ we consider the \emph{weight} (i.e.\ measure in the domain of $\Omega_U$) which has \emph{come in}, as well as the weight which has \emph{gone out}, since stage $s_{i}$. 
 The weight that has come in consists of the weight added into
 $\Omega_U$ during the stages in $(s_i, s_{i+1}]$, on numbers in $X_{s_{i+1}}$. So the weight that has come in is $\sum 2^{-|\sigma|}$, where the sum is over all $\sigma$ for which we have seen new computations $U(\sigma)\downarrow =n$ for $n\in X_{s_{i+1}}$. The weight that goes
out during the same interval, consists of the weight accumulated in $\Omega_U$ on numbers in 
$X_{s_{i}}-X_{s_{i+1}}$, \ie $\Omega_U(\{n\})[s_i]$ for numbers $n$ that were removed from the $\Pi^0_1$ approximation to $X$ during the 
stages in $(s_i, s_{i+1}]$. Note that $\omega_{i}-\omega_{i+1}$ is exactly the difference between the
weight that has gone out and the weight that has come in.  In particular, {\em the weight that comes in is at most the weight that goes out}. 

The initial weight  $\Omega_U(X)[s_0]$, as well as the weight that comes in at  later stages $s_{i+1}$,
will be dynamically split into \emph{quanta} (namely rational amounts that add up to the total weight in question) 
and placed onto various levels of a \emph{decanter} whose levels extend downward to infinity. 
Each quantum will be part of $\Omega_U(\{n\})$ for some $n\in\Nat$, and in this case we shall say
that it has label $n$. The weight in $\Omega_U(\{n\})[s_i]$ will, however,  be split into potentially many different
quanta, which are distributed at various levels of the decanter at stage $s_i$.
The basic rule for the creation of quanta corresponding to additional weight appearing in $\Omega_U$, is that
each quantum has a unique label, \ie it does not correspond to weight which is partly
in $\Omega_U(\{n\})$ for some $n$ and partly in $\Omega_U(\{m\})$ for some $m\neq n$.

At stage $s_0$, for each $n\in X_{s_0}$ such that 
$\Omega_U(\{n\})[s_0]>0$ we place a quantum of weight $\Omega_U(\{n\})[s_0]$ on the zero level
of the decanter, which we call $\CC_0$. Note that, by  standard conventions, for each 
$n$ such that $\Omega_U(\{n\})[s_0]>0$ we have $n<s_0$, so this is a finite and effective operation
and only finitely many quanta are added into $\CC_0$.
Each of the following stages $s_{i+1}$ consists of two operations (to be defined precisely later):
\begin{enumerate}[\hspace{0.5cm}(a)]
\item \emph{Purge} the quanta that have labels in $X_{s_{i}}-X_{s_{i+1}}$.
\item Split the weight that has come in into quanta, and distribute these quanta to the various levels of the decanter.
\end{enumerate}
These two operations will be carried out so as to  create the effect of quanta flowing down the levels of the decanter, possibly with some loss
of weight as they travel from each level to the one below. Although this is the most convenient way to
look at the construction, the reader should not forget that, in reality, when a quantum (or a part of it) flows
from one level to the next one, the quantum in the next level corresponds to different weight in $\Omega_U$
than the quantum at the previous level (and will have a different label). In other words, there is a resistance to
quanta flowing from level to level, due to the fact that each flow corresponds to additional weight in
$\Omega_U$. This is the crucial observation upon which we will build a \ml test for $\Omega_U(X)$.

%\paragraph{\bf Construction of the decanter flow.}
\subsection{Construction of the decanter flow}
At each stage $s_{i+1}$ we consider the numbers in $X_{s_{i}}-X_{s_{i+1}}$ and {\em purge} all the
quanta in the decanter which have their label in $X_{s_{i}}-X_{s_{i+1}}$. According to the basic rule for the creation of quanta described above, each
quantum has a unique label so there is no ambiguity in this selection. Each of these quanta can be thought of as a pair $(n, \delta)$, where $n$ is the label and $\delta$ is the weight of the quantum. To say that a quantum sitting at level $i$ is purged, simply means that we no longer consider that quantum to sit at level $i$. 

The basic idea now is that we wish to place quanta corresponding to the weight that has come in during the interval $(s_i,s_{i+1}]$, into various levels of the decanter. We wish to do this in such a way that each quantum placed can be thought of replacing or \emph{flowing from} one of the quanta (or part of one of the quanta) just purged, being placed in the decanter one level below the quantum it flows from. The fact that the weight coming in is at most the weight going out means that the purged quanta will have sufficient weight for us to be able to carry out such a procedure. Let the weight that has come in be given as a finite list  $v_1,\dots,v_r$ where each $v_i$ is a pair $(n_i,\delta_i)$ such that $\delta_i\neq 0$ is the new weight that has appeared in $\Omega_U(\{n_i \})$. Similarly enumerate the quanta just purged as $u_1,\dots,u_l$ where each $u_i$ is a pair $(m_i,\rho_i)$ which has just been purged from level $\ell_i$. Define $k_1=1$ and $x_1=0$. We define and place the new quanta in $r$ steps. At step $t$ with $1\leq t \leq r$, we are given $k_t$ which is the least $i$ such that  $u_i$ is still available for use in creating new quanta, as well as $x_t$ which is the weight of $u_{k_t}$ which has already been used. We create new quanta corresponding to $v_t$ as follows. Redefine $\rho_{k_t}$ to be $\rho_{k_t}-x_t$ (to account for the weight from this quantum already used). Let $k_{t+1}$ be the least such that $\sum_{i=k_t}^{k_{t+1}} \rho_i \geq \delta_t$ and let $x_{t+1}$ be such that $(\sum_{i\in [k_t,k_{t+1})} \rho_i) +x_{t+1} =\delta_t$. For each $i\in [k_t,k_{t+1})$ place a quantum $(n_t,\rho_i)$ at level $\CC_{\ell_i+1}$ and say that this quantum \emph{flows from} $u_i$. Also place a quantum $(n_t,x_{t+1})$ at level $\CC_{\ell_{k_{t+1}}+1}$ and say that this quantum flows from $u_{k_{t+1}}$. 

Once step $r$ is completed and all of the new quanta have been created and placed at stage $s_{i+1}$, there remains a set of quanta just purged, which were not used in creating new quanta at this stage. More precisely, consider $k_{r+1}$ and $x_{r+1}$ (according to the notation above). Redefine $u_{k_{r+1}}=(m_{k_{r+1}},\rho_{k_{r+1}}-x_{r+1})$. We consider the quanta $u_{k_{r+1}},\dots, u_l$ to be \emph{evicted} at stage $s_{i+1}$. This concludes the description of the construction.

%\paragraph{\bf Properties of the decanter flow.}
\subsection{Properties of the decanter flow}
As we explained above, the construction can be visualised as a dynamic flow of quanta through
a decanter with a  countably infinite number of levels $\CC_i$,
extending downwards and starting from level $\CC_0$.

\begin{lem}[Finiteness of throughput]\label{mOgXvfMWGJ}
The number of quanta ever enumerated into each level is finite.
\end{lem}
\begin{proof}
We use induction on the levels of the decanter. By construction at each stage only finitely
many quanta are enumerated into each level. In fact,  no enumerations into
$\CC_0$ occur except at stage $s_0$. 
Inductively assume that there is a stage $u_{0}$ after which no enumeration occurs in
$\CC_{i}, i<m$. After stage $u_{0}$ there are only finitely many stages where some quantum
in $\CC_{i}, i<m$ is purged. So there exists a stage $u_{1}>u_0$ after which no
quantum in $\CC_{i}, i<m$ is purged. By construction, this means that 
no enumeration of quanta into $\CC_m$ will occur after stage $u_1$. This concludes the induction step and the proof of the lemma.
\end{proof}

\begin{lem}[Weight bound of total throughput]\label{omQJNhsB2D}
Let $k\in\Nat$. The sum of the weights of all quanta ever enumerated into $\CC_k$, is less than $1/(k+1)$.
\end{lem}
\begin{proof}
This is trivially true for $k=0$, because $\Omega_U<1$. For $k=1$ note that 
every quantum that is enumerated into $\CC_1$ at some stage $s_{i+1}$ 
flows from part of a quantum 
in $\CC_0$ of equal weight that has just been purged during the first step of stage $s_{i+1}$.
Moreover the weight in $\Omega_U$ that enters $\CC_0$ is disjoint from the weight that enters $\CC_1$ (i.e.\ corresponds to different convergent computations for $U$).
This means that the total weight entering $\CC_1$ can be at most $\Omega_U/2<\frac{1}{2}$. Similarly, the quanta entering $\CC_2$ flow from quanta of equal weight in $\CC_1$, but once again the weight in $\Omega_U$ enumerated into $\CC_2$ is disjoint from the weight in $\Omega_U$  enumerated into $\CC_0$ and $\CC_1$.  So if $q$ is the total weight of quanta enumerated into  $\CC_2$, then 
$3q\leq \Omega_U$, which means that $q<1/3$. In the same way, for each $k$,  
the weight of the quanta enumerated into $\CC_k$  is less than $1/(k+1)$.
\end{proof}

%\paragraph{\bf Description of the \ml test.}
\subsection{The \ml test}
The basic idea is that the $i$th member $U_i$ of the \ml test will be defined according to the
enumeration of quanta into $\CC_i$. The upper bound of Lemma \ref{omQJNhsB2D}
will be instrumental in giving an upper bound for the measure of $U_i$. The enumeration of $U_i$ 
can be split into different cycles, such that during each cycle there is no eviction 
 of quanta from $\CC_j, j<i$. Upon such an eviction  we regard
the current cycle  as {\em injured}, we terminate it (preventing any further enumeration of 
intervals into $U_i$ this very stage) and then we begin a new cycle
at the next stage. By Lemma \ref{mOgXvfMWGJ} there will be finitely many injuries of $U_i$, and the enumeration of $U_i$ will have a final cycle which runs indefinitely. When counting the measure of reals  enumerated
into $U_i$, of course we also need to consider the reals  enumerated into $U_i$ during 
cycles that are later injured. Each cycle of $U_i$ has a parameter $\delta_i$ which is the length of the
intervals enumerated into $U_i$ in response to the approximation $(\omega_s)$.
For the $k$-th cycle of $U_i$ we define $\delta_i=2^{-i-k-2}$. A stage is called $i$-valid if no
eviction of quanta occurs in $\CC_j, j<i$. 

\paragraph{\bf Construction of the \ml test.}
At each $i$-valid stage $s$ 
we enumerate the interval $(\omega_s-\delta_i[s],\omega_s)$ into $U_i$. At each stage which is not  $i$-valid, 
we say that $U_i$ is injured, we redefine $\delta_i:=\delta_i/2$ and go to the next stage.
This completes the definition of $U_i$.

\paragraph{\bf Properties of the \ml test.}
Note that the intervals enumerated into $U_i$ are typically overlapping. Since $U_i$ is injured only finitely many times, $\delta_i$ reaches a limit, meaning that $\Omega_U(X)$ belongs to co-finitely many of the intervals enumerated
into $U_i$. It remains to obtain an appropriate upper bound for the measure of $U_i$.
Let $U_i(k)$ consist of the intervals enumerated into $U_i$ during its $k$th cycle.
Then the measure of $U_i(k)$ is the measure of the initial interval, which is $2^{-i-k-2}$, plus
the weight of all the quanta evicted during the stages of the $k$th cycle. Note that
these are all quanta that have passed through $\CC_i$, and which cannot be evicted again at later stages or during later cycles (a quantum can only be evicted once, and once evicted no quanta flow from it). 
By Lemma \ref{omQJNhsB2D}
the total weight of these evicted quanta over all cycles is bounded by $1/(i+1)$.  We can conclude that the weight of $U_i$
is bounded above by: 
\[
\sum_j 2^{-i-j-2} + 1/(i+1)<2/(i+1).
\]
Defining $V_i=U_{2^{i+1}}$ we have that $(V_i)$ is a \ml test such that $\cap_i V_i$ 
contains $\Omega_U(X)$.
So $\Omega_U(X)$ is not \ml random, which concludes the proof of Lemma \ref{wd9LMyiLha}.

%\bibliographystyle{alpha}
%\bibliography{omegax}
\end{document}